\documentclass{article}

\usepackage{PRIMEarxiv}

\usepackage[utf8]{inputenc} 
\usepackage[T1]{fontenc}    
\usepackage{hyperref}       
\usepackage{url}            
\usepackage{booktabs}       
\usepackage{amsfonts}       
\usepackage{nicefrac}       
\usepackage{microtype}      
\usepackage{lipsum}
\usepackage{fancyhdr}       
\usepackage{graphicx}       
\graphicspath{{media/}}     

\usepackage{amsmath,amssymb}
\usepackage{algorithmic}
\usepackage{textcomp}
\usepackage{xcolor}

\usepackage{tablefootnote}

\usepackage{amsthm}    
 \newtheorem{theorem}{Theorem}[section]
 
 \newtheorem{proposition}{Proposition}[theorem]
  \newtheorem{remark}{Remark}[theorem]
 \newtheorem{lemma}[theorem]{Lemma}
 \theoremstyle{definition}
 \newtheorem{definition}{Definition}[section]

\usepackage{glossaries}
\newacronym{aes}{AES}{Advanced Encryption Standard}
\newacronym{aqc}{AQC}{Adiabatic Quantum Computing}
\newacronym{cnf}{CNF}{Conjunctive Normal Form}
\newacronym{crc}{CRC}{Cyclic Redundancy Check }
\newacronym{dnf}{DNF}{Disjunctive Normal Form}
\newacronym{mppk}{MPPK}{Multivariate Polynomial Public Key}
\newacronym{ecc}{ECC}{Elliptic Curve Cryptography}
\newacronym{eke}{EKE}{Encrypted Key Exchange}
\newacronym{pfp}{PFP}{Factoring Polynomial Problem}
\newacronym{grh}{GRH}{Generalized Riemann Hypothesis}
\newacronym{IND-CPA}{IND-CPA}{Indistinguishability Under Chosen-Plaintext Attack}
\newacronym{IND-CCA2}{IND-CCA2}{Indistinguishability Under Adaptive Chosen Ciphertext Attack}
\newacronym{kem}{KEM}{Key Encapsulation Mechanism}
\newacronym{lwe}{LWE}{Learning With Errors}
\newacronym{lwr}{LWR}{Learning With Rounding}
\newacronym{mpkc}{MPKC}{Multivariate Public Key Cryptosystem}
\newacronym{nist}{NIST}{National Institute of Standards and Technology}
\newacronym{ntru}{NTRU}{Nth degree Truncated polynomial Ring Units}
\newacronym{pke}{PKE}{Public-Key encryption}
\newacronym{pki}{PKI}{Public Key Infrastructure}
\newacronym{pqc}{PQC}{Post-Quantum Cryptography}
\newacronym{rsa}{RSA}{Rivest-Shamir-Adleman}
\newacronym{sidh}{SIDH}{Super singular Isogeny Diffie-Hellman} 
\newacronym{svp}{SVP}{Shortest Vector Problem}
\newacronym{tls}{TLS}{Transport Layer Security}

\title{A New Symmetric Homomorphic Functional Encryption over a Hidden Ring for Polynomial Public Key Encapsulations
}

\author{
  Randy Kuang, Maria Perepechaenko, Ryan Toth \\
  Quantropi Inc. \\
  Ottawa, Canada\\
  \texttt{\{randy.kuang, maria.perepechaenko, ryan.toth\}@quantropi.com} \\
}

\begin{document}
\maketitle

\begin{abstract}
This paper proposes a new homomorphic functional encryption using modular multiplications over a hidden ring. Unlike traditional homomorphic encryption where users can only passively perform ciphertext addition or multiplication, the homomorphic functional encryption 
retains homomorphic addition and scalar multiplication properties, but also allows for the user's inputs through polynomial variables.  
The homomorphic encryption key consists of a pair of values, one used to create the hidden ring and the other taken from this ring to form an encryption operator for modular multiplication encryption. The proposed homomorphic encryption can be applied to any polynomials over a finite field, with their coefficients considered as their privacy. 
We denote the polynomials before homomorphic encryption as plain polynomials and after homomorphic encryption as cipher polynomials. A cipher polynomial can be evaluated with variables from the finite field, GF(p), by calculating the monomials of variables modulo a prime p. 
These properties allow functional homomorphic encryption to be used for public key encryption of certain asymmetric cryptosystems, such as 
Multivariate Public Key Cryptography schemes or MPKC to hide the structure of its central map construction. We propose a new variant of MPKC with homomorphic encryption of its public key. This variant simplifies MPKC central map to two multivariate polynomials constructed from polynomial multiplications, applying homomorphic encryption to the map, and changing its decryption from employing inverse maps to a polynomial division. 
We propose to 
use a single plaintext vector and a noise vector of multiple variables to be associated with the central map, in place of the secret plaintext vector to be encrypted in MPKC. We call this variant of encrypted MPKC, a Homomorphic Polynomial Public Key algorithm or HPPK algorithm. The HPPK algorithm holds the property of indistinguishability under the chosen-plaintext attacks or IND-CPA. The overall classical complexity to crack the HPPK algorithm is exponential in the size of the prime field GF(p). We briefly report on benchmarking performance results using the SUPERCOP toolkit. Benchmarking results demonstrate that HPPK offers rather fast performance, which is comparable and in some cases outperforms the NIST PQC finalists for key generation, encryption, and decryption.  
\end{abstract}

\keywords{Homomorphic Functional Encryption \and Post-Quantum Cryptography \and Public-Key Cryptography \and PQC \and Key Encapsulation Mechanism \and KEM \and Multivariate Public Key Cryptosystem \and MPKC \and PQC Performance.}

\section{Introduction}
Homomorphic encryption was first proposed by Rivest et al. in 1978 \cite{Rivest1978-HE}, one year after filing the patent for the RSA public key cryptography~\cite{Rivest1977}. Homomorphic encryption commonly refers to privacy encryption for computation in an encrypted mode, without knowing the homomorphic key and the decryption procedure. This is noticeably different from the cryptographic algorithms used to encrypt data for secure communications or storage with public key mechanisms such as RSA \cite{Rivest1977} and Diffie-Hellman \cite{Diffie1976}, and Elliptic Curve Cryptography  \cite{Koblitz1987,Miller1986} to establish the shared key for symmetric encryption using algorithms as Advanced Encryption Standard or AES.

Homomorphic encryption can be classified into partially homomorphic and fully homomorphic. The partially homomorphic encryption supports either multiplicative or additive homomorphic operations, RSA \cite{rivest1978-RSA} and ElGamal cryptosystems \cite{Elgamal85apublic} are multiplicatively homomorphic; Goldwasser–Micali \cite{HE-goldwasser82}, Benaloh \cite{HE-benaloh-2011}, and Paillier \cite{HE-Paillier-1999} are additively homomorphic. The first milestone for fully homomorphic encryption was achieved by Gentry in 2009 using lattice-based cryptography \cite{FHE-craig-2009} to support both addition and multiplication operators in the encrypted mode. Meanwhile, Chan in 2009 proposed a symmetric homomorphic scheme based on improved Hill Cipher \cite{Chan-symmetric-Hill-2009}. Kipnis and Hibshoosh proposed their symmetric homomorphic scheme in 2012 \cite{Kipnis12efficientmethods} with a randomization function for non-deterministic encryption. Gupta and Sharma proposed their symmetric homomorphic scheme based on linear algebraic computation in 2013 \cite{HE-symmetric-Gupta-2013}. Very recently, Li et al. in 2016 proposed a new symmetric homomorphic scheme, called Li-Scheme for outsourcing databases \cite{HE-Symmetric-Li-2016}. Their scheme, at large, is associated with two finite fields: a secret small field $\mathbb{F}_q$ and a big public field $\mathbb{F}_p$, with modular exponentiation with its secret base $s$ followed by modular multiplication with plaintext message $m$. Li-Scheme supports both additive and multiplicative operations so it is a full homomorphic encryption. Wang et al. performed a cryptoanalysis of the Li-Scheme in 2018 \cite{HE-symmetric-Wang-2018} and broke the scheme with certain known plaintext-ciphertext pairs. Wang et al. further improved their cryptoanalysis in 2019 and successfully recovered the secret key with the ciphertext-only attack using lattice reduction algorithm \cite{HE-Symmetric-Qu-2019}.

Homomorphic encryption solely focuses on addition and multiplication operations on the encrypted data for plain data privacy. 
However, it would be very interesting to see an extension of homomorphic encryption from data privacy to functional privacy with variables to take the user's inputs in a framework of a public key scheme. It is rarely seen that a public key cryptosystem, more precisely quantum-safe public key cryptosystem, is purposely designed with careful considerations not only to leverage homomorphic properties of ciphertext addition and scalar multiplication but also to take user's secrets into ciphertext computation. This served as a motivation for our paper. We introduce a new Homomorphic Polynomial Public Key encapsulation or HPPK, which is an asymmetric key encapsulation scheme, with public keys encrypted using functional homomorphic encryption. HPPK uses multivariate polynomials to not only leverage homomorphic properties of addition and scalar multiplication but also allows for encrypting party's input during the ciphertext creation. That is, the public key polynomial coefficients are encrypted using homomorphic function to ensure they are never truly public and hide the structure of the public key, at the same time, treating variables of the said public key polynomials as user input allows for freedom during the encryption process.

HPPK cryptosystem has two distinct features, namely, the homomorphic encryption of the public key that allows for the user's input during ciphertext creation, and the use of a hidden ring. HPPK is not the first cryptosystem to use hidden structure. For instance, the work of Li \textit{et al.} describes a cryptosystem with a hidden prime ring~\cite{HE-Symmetric-Li-2016}. Another important example is Hidden Field Equations (HFE) cryptosystems. The examples of asymmetric multivariate encryption schemes that are based on HFE include~\cite{Kipnis99unbalancedoil,Square,ZHFE,SzepieniecAlan2016EFCA}. Various signature schemes based on HFE were also proposed~\cite{PATARINJacques2001Q1ld,PetzoldtAlbrecht2016DPfH,ZhangTan, Ding2005RainbowAN, Casanova2017GeMSSAG}. In the framework of HFE, the private polynomials as well as the structure they are defined over, a field extension, are both hidden using affine transformations.


The algorithms based on HFE, mentioned above fall in the category of quantum-safe algorithms. More precisely, multivariate quantum-safe algorithms. Quantum computing developments have been receiving a lot of focus from the academic community as well as industry leaders since Google announced its first quantum advantage in 2019 \cite{Arute2019}. But it was the National Institution of Standards and Technology (NIST) that opened the arena for quantum-resistant cryptography, when they started the post-quantum cryptography (PQC) standardization process in November 2017. Recently, they have announced third-round finalists which include four key exchange mechanism schemes (KEM) and three finalists for digital signatures \cite{NIST2021}. Four KEM finalists include code-based Classic McEliece \cite{McEliece1978}, lattice-based CRYSTALS-KYBER \cite{avanzi2017crystals}, NTRU \cite{Hoffstein1998,Bernstein2018}, and Saber \cite{vercauterensaber} algorithms. At the latest announcement, NIST selected CRYSTALS-KYBER to be standardized algorithm for KEM. In addition to the aforementioned finalists for KEM, the Multivariate Public Key Cryptosystems or MPKC is worth a special discussion. Algorithms based on multivariate polynomial problems are considered to be quantum-safe, but they also make an excellent candidate for homomorphic encryption due to the use of multivariate polynomials. 

The framework of MPKC is built on a system of quadratic polynomials. The public key is represented by a central map $\mathcal{P}: \mathbb{F}_p^m \rightarrow \mathbb{F}_p^{\ell}$ with $m$ variables and $\ell$ polynomials \cite{Ding2009mpkc}. Many variants of MPKC central map constructions have been proposed since Matsumoto and Imai first introduced this cryptosystem in 1988 \cite{mpkc1-1988}, including single field systems and mixed field systems \cite{mpkc-wolf}. Single field MPKC includes several Triangular systems and the Oil and Vinegar system since Patarin and Goubin in 1997 \cite{Patarin97trapdoorone-way} and unbalanced Oil and Vinegar scheme by Kipnis et. al. in 1999 \cite{Kipnis99unbalancedoil}. The mixed field MPKC refers to Matsumoto-Imai system \cite{mpkc1-1988} and Hidden Field Equation \cite{HFE-1996}. In addition, Wang et al. in 2006 proposed a Medium-Field MPKC scheme \cite{Wang_amediumeld} and an improved scheme in 2008~\cite{mdeium-field-mpkc-2008}. Ding and Schmidt proposed Rainbow as a MPKC digital signature scheme in 2005~\cite{Ding2005RainbowAN}. 

Attacks on MPKC cryptosystems are mainly classified into two categories:  algebraic solving attacks and linear algebra attacks.  Algebra solving attacks attempt to solve the MPKC multivariate equation system from the public key with ciphertext $(z_1, z_2, \dots, z_{\ell})$ to recover the pre-image $(x_1, x_2, \dots, x_m)$.  Faugére reported his first attack on MPKC in 1999\cite{FAU99}  and in 2002\cite{Fau02} using Gröber bases ($F_4$), later in 2003 Faugére and Antoine reported their attack on HFE Gröber bases ($F_5$). Ding et. al. proposed their new Zhuang-Zi algorithm to solve the multivariate system in 2006 \cite{Ding2006ZhuangZiAN}. In linear algebra attacks, Courtois \textit{et al.} reported their attack on MPKC using the relinearization technique, aclled XL in 2000 \cite{XL-2000}.  The Minrank attack has been successfully applied by Goubin and Courtois on the single field system in 2000 \cite{minrank-goubin-200} and by Kipnis and Patarin on the mixed field system in 1999 \cite{minrank-kipnis99}.

A new type of polynomial public key has been recently proposed by Kuang in 2021~\cite{kuang2021ACCC}, based on univariate polynomial multiplications, by Kuang and Barbeau in 2021~\cite{kuang2021CCECE, kuang2021performance, kuang2021QCE21}, based on multivariate polynomial multiplications with two noise functions to increase the public key security against possible public key attacks. The digital signature scheme of the multivariate polynomial public key or MPPK has been prosoed by Kuang, Perepechaenko and Barbeau in 2022~\cite{kuang2022.08.01-DS}. This paper explores the possibility to combine a new homomorphic encryption to key construction to further enhance the security of the MPPK cryptography for key encapsulation mechanism or KEM.

The proposed HPPK scheme can be also considered as a new variant of MPKC scheme with public keys being encrypted using homomorphic functional encryption. We begin by introducing the proposed symmetric homomorphic encryption scheme in~\nameref{sec:SHE}. The proposed HPPK algorithm is then discussed in~\nameref{sec:HPPK}. We present the reader with thorough security analysis of HPPK in~\nameref{sec:security}, and report on benchmarking the performance of HPPK in~\nameref{sec:bench}. We conclude with~\nameref{conclusion}.

\section{A New Symmetric Homomorphic Functional Encryption over a Hidden Ring}\label{sec:SHE}

In contrast to conventional homomorphic cryptography used for data privacy, in this paper we propose homomorphic functional encryption to be applied to the public key in the framework of multivariate asymmetric cryptography. This will allow for an asymmetric scheme with encrypted public keys. Moreover, by construction, functional homomorphic encryption allows for user's input during the ciphertext generation procedure. That is, the ciphertext can be created with the input of the encrypting party, however, the public key used for encryption is itself encrypted using functional homomorphic operator. The decrypting party is the only party that has knowledge of the private key associated with the functional homomorphic encryption operator as well as the asymmetric scheme private key. Essentially, the  homomorphic functional encryption defined in this paper provides a round-trip envelope for a public key encryption. 

In a way, such approach combines three main areas of cryptography, namely, asymmetric cryptography, homomorphic encryption, and symmetric cryptography with self-shared key. This phenomenon is illustrated in Fig.~\ref{scheme}. In the figure, the traditional public key derived from a given assymetric algorithm is called plain public key or PPK, the homomorphically encrypted PPK is called cipher public key or CPK. The cipher is produced by evaluating the public key polynomial values using a user-selected secret. The decryption would perform in two stages: homomorphic decryption and then secret extraction.

\begin{figure}[h]
\caption{Illustration of a cryptosystem combining asymmetric cryptography, symmetric cryptography with a single self-shared key, and homomorphic encryption. }
\includegraphics[width=6.5cm, height=5.5cm]{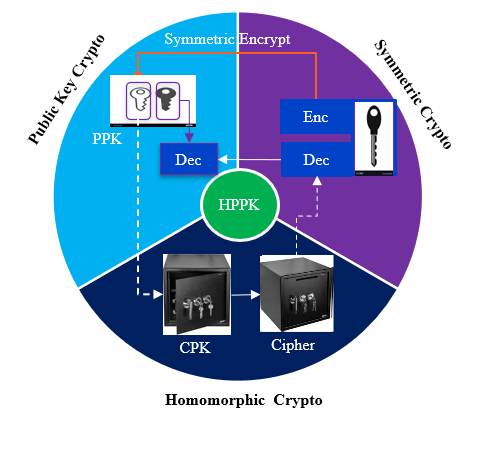}
\centering
\label{scheme}
\end{figure}

We begin by introducing the Homomorphic Functional Encryption Operator. In order to allow for the user's input during ciphertext creation, and leverage additive and scalar multiplicative homomorphic features, the functional homomorphic encryption is applied to polynomials. We discuss the reason for this further. Hence, when introducing the said operator we assume that it will be applied to polynomials.


\subsection{Homomorphic Encryption Operator}
Let $S$ be a positive integer, and $R$ be a randomly chosen value such that $R \in \mathbb{Z}_{S}$ and $gcd(R, S) =1.$ We propose a Homomorphic Functional Encryption Operator $\hat{\mathcal{E}}_{(R, S)}$, with a secret homomorphic key being a tuple $(R,S)$. The values $S$ and $R$ are never shared. 

In its general form, the encryption operator is defined as a multiplicative operation modulo a hidden value $S$ as
\begin{equation}\label{eq:he}
    \hat{\mathcal{E}}_{(R, S)}(f) = (R \circ f) \ \text{mod} \  S,
\end{equation}
where $f$ denotes 
any univariate or multivariate polynomial $f = \sum_{i = 0}^{k}f_{i}X_{i}$, over $\mathbb{F}_p$ with $X_i$ being its monomials. The encryption operator acts on the coefficients of $f$, which we refer to as plain coefficients. 
This produces what we call cipher coefficients, denoted $h_i$, as 
\begin{equation*}
    \hat{\mathcal{E}}_{(R, S)}f_i = Rf_i \ \text{mod} \  S = h_i
\end{equation*}
for every $i$. Similarly, we can define a Homomorphic Functional Decryption Operator as 
\begin{equation}\label{eq:hd}
    \hat{\mathcal{E}}_{(R^{-1}, S)}h = (R^{-1} \circ  h) \ \text{mod} \  S.
\end{equation}
Such operator decrypts the coefficients of the polynomial $h$. That is, it successfully decrypts the cipher coefficients back to the plain coefficients. More precisely, 
\begin{equation*}
\hat{\mathcal{E}}_{(R^{-1}, S)}h_i = R^{-1}(Rf_i) \ \text{mod} \  S = (R^{-1}R)f_i \ \text{mod} \  S = f_i
\end{equation*} for any $i.$

The above defined homomorphic operator $\hat{\mathcal{E}}_{(R, S)}$ holds following homomorphic properties:
\begin{itemize}
    \item [$\bullet$] $\hat{\mathcal{E}}_{(R, S)}$ is additively homomorphic: if $a$ and $b$ are two plain constants, then $\hat{\mathcal{E}}_{(R, S)} (a + b) = R(a+b)  \text{ mod }S= Ra + Rb \text{ mod }S = \hat{\mathcal{E}}_{(R, S)} a + \hat{\mathcal{E}}_{(R, S)} b$;
    \item [$\bullet$] $\hat{\mathcal{E}}_{(R, S)}$ is scalar multiplicatively homomorphic: if $a$ is a plain constant and $x$ is a variable, then $\hat{\mathcal{E}}_{(R, S)} (ax) = R(ax) = (Ra)x = [\hat{\mathcal{E}}_{(R, S)}a] x$.
\end{itemize}
Thus, the operator $\hat{\mathcal{E}}_{(R, S)}$ offers partially homomorphic encryption. We leave it to the reader to verify that the same properties hold true for the proposed homomorphic functional decryption operator $\hat{\mathcal{E}}_{(R^{-1}, S)}$.  
Note that these homomorphic properties come from linearity, and thus are natural to polynomials. Indeed, polynomials hold additive and scalar multiplicative properties through their coefficients. Moreover, polynomials can be defined and evaluated with coefficients in a field or a ring, different from a field or a ring for variables. We leverage this property, and thus, apply the functional homomorphic encryption to public key cryptosystems with polynomial public keys. 



\subsection{Homomorphically Encrypted Polynomials}
As we have previously stated, the proposed homomorphic encryption is applicable to all polynomials over a ring $\mathbb{Z}_p$ or finite field $\mathbb{F}_p$ characterized by a prime $p$. In this work, when we refer to polynomials, we imply that the plain polynomials, to be encrypted, are considered modulo $p$, unless stated otherwise. A generic multivariate polynomial has the following form
\begin{equation*}
p(x_1,\ldots,x_m)
=
\sum_{j_1=1}^{\ell_1} 
\cdots 
\sum_{j_m=1}^{\ell_m}
p_{ij_1\ldots j_m} x_1^{j_1} \cdots x_m^{j_m}.
\end{equation*}
Alternatively, let $X_j=x_1^{j_1} \cdots x_m^{j_m}$ denote the monomials of such polynomial, then

\begin{equation}\label{eq:genericP}
p(x_1,\ldots,x_m) = \sum_{j=1}^{L} p_{j} X_j,
\end{equation}
where L denotes the total number of terms. 

To successfully encrypt and decrypt 
any polynomial $p(x_1, \dots, x_m)$ using the functional homomorphic encryption and decryption operators defined in Eq.~\eqref{eq:he} and~\eqref{eq:hd} respectively, the following conditions must be met:
\begin{itemize}
    \item [$\bullet$] The monomials $X_j$ are to be computed as $X_j = (X_j \ \text{mod} \ p)$. The values of monomials reduced modulo $p$ are used to compute the value of the polynomial $p(x_1, \dots, x_m)$. 
    \item [$\bullet$]  The homomorphic secret key value $S$ should satisfy the bit length condition: $|S|_2 > 2|p|_2 + |L|_2$.
\end{itemize}
\sloppy
The first condition ensures that polynomial $p(x_1, \dots, x_m)$ is evaluated as if the mononials $X_j$ are new variables over $\mathbb{F}_p$, 
Indeed, the operator $\hat{\mathcal{E}}_{(R, S)}$ is applied to the polynomial $p(x_1, \dots, x_m)$ in the following way 
\begin{equation}
\hat{\mathcal{E}}_{(R, S)}p(x_1, \dots, x_m) = \sum_{j=1}^{L} [Rp_{j} \text{ mod }S] X_j.
\end{equation}
Such encrypted polynomial can be computed as $$\sum_{j=1}^{L} [Rp_{j} \text{ mod }S](X_j \text{ mod p}) = \bar{p}.$$ Note that the computed value was not reduced modulo any integer, nor is the arithmetic performed modulo any integer. Thus, the user's input through monomials $X_j$ remains intact and can be decrypted correctly. Let the plain value of the polynomial with user's input, that is, if the polynomial was not encrypted with $\hat{\mathcal{E}}_{(R, S)},$ be $$\hat{p} = \sum_{j=1}^{L} p_{j}X_j \text{ mod p}.$$ To ensure successful decryption, the second condition must be met. If the size of $S$ is sufficiently large, the values of coefficients and variables remains the same after decryption, and it is possible to recover $\hat{p}.$ Indeed, 
\begin{equation}
\hat{\mathcal{E}}_{(R^{-1}, S)}\bar{p}  = \sum_{j=1}^{L} [R^{-1}Rp_{j} \text{ mod }S] (X_j \text{ mod p}) = \sum_{j=1}^{L}p_{j}(X_j \text{ mod p}),
\end{equation}
and then the value $\sum_{j=1}^{L}p_{j}(X_j \text{ mod p})$ can be reduced modulo $p$ to yield $\sum_{j=1}^{L}p_{j}X_j \text{ mod p} = \hat{p}.$

To elaborate more on this, we present the reader with two examples of functional homomorphic encryption of linear and quadratic polynomials.

\subsubsection{Linear Polynomials}

Recall, that we encrypt the coefficients of the polynomials defined over $\mathbb{F}_p$, which successfully maps polynomials from $\mathbb{F}_{p}[x_1, \dots, x_m]$ to $\mathbb{Z}_{S}[x_1, \dots, x_m]$, leaving $x_1, \dots, x_m \in \mathbb{F}_p.$ A generic linear multivariate polynomial over a finite field $\mathbb{F}_p$ has form
\begin{equation}\label{eq:plainL}
    p( x_1, x_2, \dots,x_m) = \sum_{j=1}^{m} p_{j} x_j \ \text{mod} \ p.
\end{equation}
Conventionally, in the asymmetric encryption schemes, the public key inherits mathematical logic from the private key, making it vulnerable. Hence, if public key consists of polynomials, we wish to encrypt the coefficients of the said polynomials using functional homomorphic operator, to hide the mathematical logic. In this case we share the cipher public key, encrypted using functional homomorphic operator. To ensure that the ciphertext can be still created in the framework of asymmetric public key scheme, the variables in the public key polynomials are used for user's input. They are not encrypted using homomorphic encryption, but only using the encryption procedure from the asymmetric scheme. Such variable values can consist of the plaintext only, or plaintext and noise used for obscurity. 

Applying homomorphic encryption operator to the above linear polynomial, defined in Eq.\eqref{eq:plainL}, produces a cipher linear polynomial with coefficients in a hidden ring $\mathbb{Z}_S$, and variables in $\mathbb{F}_p:$
\begin{align}\label{eq:cipher1}
&\mathcal{P}( x_1, x_2, \dots,x_m) = \hat{\mathcal{E}}_{(R, S)} p( x_1, x_2, \dots,x_m) \nonumber \\
&= \sum_{j=1}^{m} (R p_{j} \ \text{mod} \ S) x_j 
=\sum_{j=1}^{m} \mathcal{P}_{j} x_j. 
\end{align}
While the plain coefficients, $p_j$, are encrypted into cipher coefficients, $\mathcal{P}_j$, the cipher polynomial $\mathcal{P}( x_1, x_2, \dots,x_m)$ can still be evaluated with a set of chosen values $r_1, \dots, r_m \in \mathbb{F}_p$ to produce value $\bar{\mathcal{P}}$ 
\begin{equation}\label{eq:cipherPL}
    \bar{\mathcal{P}} =\mathcal{P}( r_1, r_2, \dots, r_m) = \sum_{j=1}^{m} \mathcal{P}_{j} (r_j \ \text{mod} \ p).
\end{equation}
Let the value $\bar{p} = \sum_{i=1}^{m} p_{j}r_j \ \text{mod} \ p,$ be the original ciphertext of the asymmetric scheme. However, it is encrypted into the value $\bar{\mathcal{P}}$ using homomorphic encryption. To recover the plain polynomial value, that is, decrypt the cipher coefficients, into the plain coefficients and evaluate polynomial modulo $p$, we first apply the functional homomorphic decryption operator $\hat{\mathcal{E}}_{(R^{-1}, S)}$ to get $ \hat{\mathcal{E}}_{(R^{-1}, S)} \bar{\mathcal{P}},$ and then reduce this value modulo $p$. More precisely,
\begin{equation*}\label{eq:nuliffy}
 \hat{\mathcal{E}}_{(R^{-1}, S)} \bar{\mathcal{P}} = \sum_{j=1}^{m} [R^{-1}\mathcal{P}_{j} \text{ mod }S] (r_j \ \text{mod} \ p) = \sum_{i=1}^{m} p_{j}(r_j \ \text{mod} \ p), 
 \end{equation*}
which reduced modulo $p$ is $$\sum_{i=1}^{m} p_{j}r_j \ \text{mod} \ p = \bar{p}.$$ 

In a framework of asymmetric scheme with functional homomorphic encryption element, polynomials such as in Eq.~\eqref{eq:plainL} are associated with plain coefficients, that is, the original public keys. The cipher polynomials have form as in Eq.~\eqref{eq:cipher1}, with coefficients being encrypted from the plain public keys, using homomorphic encryption. Such cipher public keys are shared, and the plain public keys are stored securely and never shared. The ciphertext in this combined algorithm is of the form as in Eq.~\eqref{eq:cipherPL}. The decrypting party first needs to decrypt the ciphertext to nullify the homomorphic encryption of the public key, as shown in Eq.~\eqref{eq:cipherPL}. Afterwards, the decryption party can perform decryption procedure that corresponds to the given asymmetric scheme.  


\subsubsection{Quadratic Polynomials}
Multivariate quadratic polynomials serve as the foundation of Multivariate Public Key Cryptosystem or MPKC\cite{Ding2009mpkc, Ding2020SME,Ding2020}. Thus, we want to pay special attention on applications of functional homomorphic encryption on multivariate quadratic polynomials. 
A general quadratic multivariate polynomial $p(x_1, x_2, \dots, x_n)$ over a finite field $\mathbb{F}_p$ has the following form
\begin{equation}\label{eq:quadplain}
    p( x_1, x_2, \dots,x_m) = \sum_{1\leq i\leq j}^{m} p_{ij} x_ix_j \ \text{mod} \ p,
\end{equation}
where the coefficients $p_{ij}$ are considered as the privacy constants for this polynomial function so they must be hidden from public. This is done by applying the functional homomorphic encryption operator to this polynomial as follows
\begin{align} \label{eq:cipherPQ}
\mathcal{P}( x_1, x_2, \dots,x_m) = \hat{\mathcal{E}}_{(R, S)} p(x_1, x_2, \dots,x_m) \\
= \sum_{1\leq i\leq j}^{m} (R p_{ij} \ \text{mod} \ S) x_ix_j
=\sum_{1\leq i\leq j}^{m} \mathcal{P}_{ij} x_ix_j. \\
\end{align}
Here, the encrypted coefficients are defined over the hidden ring $\mathbb{Z}_S$, however, all the variables $x_1, \dots, x_m$ are still elements of the field $\mathbb{F}_p$. As we have previously mentioned, we refer to the coefficients $p_{ij}$ as plain coefficients, and $\mathcal{P}_{ij}$ are referred to as cipher coefficients. Similarly, $\mathcal{P}( x_1, x_2, \dots,x_m)$ and $p( x_1, x_2, \dots,x_m)$ are referred to as cipher and plain polynomials respectively. 

While coefficients are encrypted with homomorphic encryption operator, the polynomial $\mathcal{P}( x_1, x_2, \dots,x_m)$ still accepts user's input. That is, the cipher polynomial value $\bar{\mathcal{P}}$ can be still calculated with a chosen set of $r_1, \dots, r_m$ from the field $\mathbb{F}_p$ as follows
\begin{equation}\label{quadcipher}
       \bar{\mathcal{P}} = \mathcal{P}( r_1, r_2, \dots, r_m) =  \sum_{1\leq i\leq j}^{m} \mathcal{P}_{ij} (x_ix_j \ \text{mod} \ p).
\end{equation}
 Note that the computed value $\bar{\mathcal{P}}$ is an integer. The arithmetic to compute such value was not performned modulo any integer. The plain polynomial values are securely hidden through the hidden ring $\mathbb{Z}_{S}$. To recover the plain polynomial equation, decryption opertor $\hat{\mathcal{E}}_{(R^{-1}, S)}$ can be applied to the cipher polynomial value $\bar{\mathcal{P}}$, followed by reduction $\text{mod} \ p $:
\begin{align}\label{eq:plainP}
&\hat{\mathcal{E}}_{(R^{-1}, S)} \bar{\mathcal{P}} = R^{-1} \bar{\mathcal{P}}\ \text{mod} \ S, \text{ then }\\
&(R^{-1} \bar{\mathcal{P}}\ \text{mod} \ S) \ \text{mod} \ p  = p( r_1, r_2, \dots,r_m) = \bar{p}.
\end{align}
The value $\bar{p}$ is the plain polynomial value for the chosen values of variables $x_1, \dots, x_m$ by the encrypting party. 

Similar to the linear case, the public key of the asymmetric scheme consist of quadratic polynomials of the form~\eqref{eq:quadplain}, to be encrypted using homomorphic functional encryption operators. The cipher public keys are of the form~\eqref{eq:cipherPQ}. Such cipher public keys are the ones shared, while the plain public keys are not. The ciphertext in the combined scheme is of the form~\eqref{quadcipher}, which needs to be decrypted back to the plain value. For that a homomorphic decryption operator is applied, as in Eq~\eqref{eq:plainP}, and the plain ciphertext value is recovered. 



\section{Homomorphic Polynomial Public Key Cryptosystem}\label{sec:HPPK}

\subsection{Brief Summary of MPKC}
An interested reader can find the detail description of MPKC schemes by Ding and Yang~\cite{Ding2009mpkc}. In this section, we briefly outline the basic mechanism of MPKC algorithms. The framework mainly consists of $\ell$ quadratic multivariate polynomials 
\begin{equation*}
   p_1(x_1, \dots, x_m), p_2(x_1, \dots, x_m), \dots, p_{\ell}(x_1, \dots, x_m)
\end{equation*}
in $m$ variables over finite field $\mathbb{F}_p$. Each polynomial $p_k(x_1, \dots, x_m)$ can be written in its expanded form as

\begin{equation}\label{eq:MQ}
       p_k( x_1, \dots,x_m) = \sum_{i<j=1}^{m} p_{ijk} x_ix_j 
\end{equation}
for $k = 1, 2, \dots, l$. In the literature Eq.\eqref{eq:MQ} is generally written in a matrix form as
\begin{equation}
    p_k( x_1,  \dots,x_m) = (x_1,  \dots,x_m) 
    \begin{pmatrix}
p_{11k} & p_{12k} & \ldots & p_{1mk}\\ 
\ldots & \ldots & \ldots & \ldots\\ 
p_{i1k} & p_{i2k} & \ldots & p_{imk}\\
\ldots & \ldots & \ldots & \ldots\\
p_{m1k} & p_{m2k} & \ldots & p_{mmk} 
\end{pmatrix}
(x_1, \dots,x_m)^T \\
\end{equation}
briefly expressed as,
\begin{equation*}
   p_k( x_1,  \dots,x_m) = \vec{x}\cdot \mathcal{P}_k \cdot \vec{x}^T,
\end{equation*}
where $\mathcal{P}_k$ is an $m\times m$ square matrix. Considering all $\ell$ polynomials, we can write the MPKC map from $\mathbb{F}_p^m$ to $\mathbb{F}_p^{\ell}$ as a 3-dimensional matrix $\mathcal{P}[\ell][m][m]$, which is a trapdoor called the central map. The central map is selected to be easily invertible. In order to protect this central map and its structure, two affine linear invertible maps $T$ and $S$ are chosen to construct the MPKC public key:
\begin{itemize}
    \item Public Key: $\mathcal{\bar{P}} = T \circ \mathcal{P} \circ S$.
    \item Private Key: $(T, \mathcal{P}, S).$
\end{itemize}

The MPKC encryption procedure simply to evaluates $\ell$ polynomials over the field $\mathbb{F}_p$ as
\begin{equation}\label{eq:mpkc-sys}
  \vec{z}= \mathcal{\bar{P}}(\vec{x}) =\{ z_1=p_1(x_1, \dots, x_m),  z_2=p_2(x_1, \dots, x_m), \dots,  z_{\ell}=p_{\ell}(x_1, \dots, x_m)\}
\end{equation}
and decryption works as follows
\begin{equation*}
    \vec{u} = T^{-1}(\vec{z}), \vec{v} = \mathcal{P}^{-1}(\vec{u}), \vec{x} = S^{-1}(\vec{v}).
\end{equation*}

The major step to use MPKC is to construct the invertible central map $\mathcal{P}$ over a finite field $\mathbb{F}_p$ to perform a map: $\mathbb{F}_p^m \rightarrow \mathbb{F}_p^{\ell}$.

There may be a potential way to enhance the security of MPKC cryptosystem by applying the proposed homomorphic encryption on its map: $\mathbb{F}_p^m\rightarrow\mathbb{F}_p^{\ell}$. The homomorphic encryption effectively hides the public key construction logic over a hidden ring $\mathbb{Z}_S$. In this case, an encryption key $R_k$ is required for each quadratic polynomial $p_k(x_1,\dots,x_m)$, with value $R_k$ chosen over the hidden ring $\mathbb{Z}_S$ for all $k$. Hence, there are a total of $\ell$ encryption keys for MPKC. The MPKC encryption in this case is almost the same as the original MPKC encryption. The ciphertext $(z_1, z_2, \dots, z_{\ell})$ is to be homomorphically decrypted to create original multivariate equation system, as illustrated in Eq.\eqref{eq:mpkc-sys}. This means, Eq.\eqref{eq:mpkc-sys} is hidden under the hidden ring  $\mathbb{Z}_S$. On one hand, applying the homomorphic encryption would increase the public key size for MPKC, however, the number of variables can be reduced due to the homomorphic encryption. 

In this paper, we are not going to further explore this variant of MPKC schemes but we will focus on another variant of MPKC, called HPPK which we propose in the new section.

\subsection{HPPK Encapsulation}
We propose a new variant of an MPKC scheme, called the Homomorphic Polynomial Public Key or HPPK, with the following  considerations:
\begin{itemize}
    \item[$\bullet$] The vector on the left hand side of the map $\mathcal{P}$ is treated as $\vec{x}_l$ and the vector on the right hand side as $\vec{x}_r$; 
    \item[$\bullet$] The vector $\vec{x}_l$ is replaced with  $\vec{x}_l =(x^0, x^1, x^2, \dots, x^n)$, considering $\vec{x}_l$ as a message vector in a polynomial vector space represented by a basis $\{x^0, x^1, x^2, \dots, x^n\}$ for a message variable $x$ and $\vec{x}_r=(x_1, \dots, x_m)$ as a noise vector for noise variables $x_1, \dots, x_m$;
    \item[$\bullet$] The proposed homomorphic encryption is applied to the central map $\mathcal{P}$, mapping the elements from $\mathbb{F}_p \rightarrow \mathbb{Z}_S$:
    \subitem $\mathcal{\bar{P}} = \hat{\mathcal{E}}_{(R, S)}\mathcal{P}$ \\
    and the decryption is de-mapping from $\mathbb{Z}_S \rightarrow \mathbb{F}_p$:
    \subitem  $\mathcal{P} = \hat{\mathcal{E}}_{(R^{-1}, S)}\mathcal{\bar{P}} \ \text{mod} \ p$
    \item[$\bullet$] The number of polynomials is reduced to $\ell = 2$;
    \item[$\bullet$] The decryption mechanism is changed from inverting maps to modular division, which automatically cancels the noise used for obscurity.
\end{itemize}

\subsubsection{Key Construction}
Without loss of generality, we change the notation of the unencrypted central map to $P$. Under the above considerations, the central map $P$ consists of two multivariate polynomials 

\begin{equation} \label{eq:p1}
    p_1(x, x_1, x_2, \dots,x_m) = (1, x^1,  \dots,x^n) 
    \begin{pmatrix}
p_{011} & p_{021} & \ldots & p_{0m1}  \\
p_{111} & p_{121} & \ldots & p_{1m1}\\
\ldots & \ldots & \ldots & \ldots\\ 
p_{i11} & p_{i21} & \ldots & p_{im1}\\
\ldots & \ldots & \ldots & \ldots\\
p_{n11} & p_{n21} & \ldots & p_{nm1} 
\end{pmatrix}
(x_1, x_2, \dots,x_m)^T, 
\end{equation}
and 
\begin{equation}\label{eq:p2}
    p_2(x, x_1, x_2, \dots,x_m) = (1, x^1,  \dots,x^n) 
    \begin{pmatrix}
p_{012} & p_{122} & \ldots & p_{0m2}  \\
p_{112} & p_{122} & \ldots & p_{1m2}\\ 
\ldots & \ldots & \ldots & \ldots\\ 
p_{i12} & p_{i22} & \ldots & p_{im2}\\
\ldots & \ldots & \ldots & \ldots\\
p_{n12} & p_{n22} & \ldots & p_{nm2} 
\end{pmatrix}
(x_1, x_2, \dots,x_m)^T.
\end{equation}
Note that the matrix maps ${P}_1$ and ${P}_2$ are of size $(n+1)\times m$, thus, no longer square. The construction of $p_1( x, x_1, \dots,x_m)$ and $p_2( x, x_1, \dots,x_m)$ can alternatively be achieved with polynomial multiplications
\begin{align}\label{eq:pbf}
    &p_1(x, x_1, \dots,x_m) = b(x, x_1, \dots,x_m)f_1(x) \\
    &\nonumber p_2(x, x_1, \dots,x_m) = b(x, x_1, \dots,x_m)f_2(x),
\end{align}
where the base multivariate polynomial $b(x, x_1, x_2, \dots,x_m)$  and univariate polynomials $f_1(x)$ and $f_2(x)$ have the following generic forms
\begin{align}\label{eq:bff}
    &b(x, x_1, \dots,x_m) = \sum_{i=0}^{n_b} \sum_{j=1}^m b_{ij}x^ix_j  \\
    &\nonumber f_1(x) = \sum_{i=0}^{\lambda} f_{1i} x^i\\
    &\nonumber f_2(x) = \sum_{i=0}^{\lambda} f_{2i} x^i. 
\end{align}
Here, $n_b$ and $\lambda$ are orders of base multivariate polynomial and univariate polynomials with respect to message variable $x$ respectively. Without loss of generality, we assume that the univariate polynomials $f_1(x)$ and $f_2(x)$ are solvable, in other words $\lambda<5$. Using Eq.\eqref{eq:pbf} and \eqref{eq:bff}, we can express
    \begin{align}\label{eq:pbff}
    &p_1(x, x_1, \dots,x_m) = \sum_{i=0}^{n}\sum_{j=1}^{m} p_{ij1}x^ix_j = \vec{x}_l \cdot P_1 \cdot \vec{x}_r \\
    &\nonumber p_2(x, x_1, \dots,x_m) = \sum_{i=0}^{n}\sum_{j=1}^{m} p_{ij2}x^ix_j = \vec{x}_l \cdot P_2 \cdot \vec{x}_r,
    \end{align}
with $p_{ij1} = \sum_{s+t=i} b_{sj}f_{1t}$ and $p_{ij2} = \sum_{s+t=i} b_{sj}f_{2t}$ being the coefficients, and $n = n_b + \lambda$. It is apparent that the plain central map $P$ as shown in Eq.\eqref{eq:p1} and Eq.\eqref{eq:p2} expanded as Eq.\eqref{eq:pbff} inherits a lot of structure. Given that the components of the central map are private key elements, the map in its unaltered form is not secure against potential attacks such as polynomial factorization, root finding, etc. 
To secure the central map, we apply functional homomorphic encryption operator to the plain central map by acting with $\hat{\mathcal{E}}_{(R_1, S)}$ on $p_1(x, x_1, \dots, x_m)$ and $\hat{\mathcal{E}}_{(R_2, S)}$ on $p_2(x, x_1, \dots, x_m)$. To be more precise, the cipher central map consists of two polynomials 
\begin{align}\label{eq:encpbff}
    &\mathcal{P}_1(x, x_1, \dots,x_m) = \sum_{i=0}^{n}\sum_{j=1}^{m} (R_1p_{ij1} \text{ mod }S)x^ix_j = \vec{x}_l \cdot \mathcal{P}_1 \cdot \vec{x}_r \\
    &\nonumber \mathcal{P}_2(x, x_1, \dots,x_m) = \sum_{i=0}^{n}\sum_{j=1}^{m} (R_2p_{ij2} \text{ mod }S)x^ix_j = \vec{x}_l \cdot \mathcal{P}_2 \cdot \vec{x}_r.
    \end{align} 
    We set public key to be the cipher central map $\mathcal{P}$, while private key consists of the homomorphic operators, the hidden ring, together with univariate polynomials: 
\begin{itemize}
    \item[$\bullet$] Security parameter: the prime finite field $\mathbb{F}_p$ which is agreed on before the key generation procedure.
    \item[$\bullet$] Private Key: 
        \subitem $\circ$ hidden ring $\mathbb{Z}_S$ with a randomly selected $S$ for the required bit length;
        \subitem $\circ$ homomorphic encryption key values $R_1$ and $R_2$ chosen from $\mathbb{Z}_S$;
        \subitem $\circ$ univariate polynomials $f_1(x)$ and $f_2(x)$ with coefficients randomly selected from $\mathbb{F}_S$;
    \item[$\bullet$] Public Key: the map $\mathcal{P}$, consisting of
        \subitem $\circ \mathcal{P}_1(\vec{x}_l, \vec{x}_r) = \vec{x}_l \cdot \mathcal{P}_1 \cdot \vec{x}_r$  
        \subitem $\circ \mathcal{P}_2(\vec{x}_l, \vec{x}_r) = \vec{x}_l \cdot \mathcal{P}_2 \cdot \vec{x}_r$ 
\end{itemize}

\subsubsection{Encryption}
Encryption is straightforward by determining the value for the secret $x$ and randomly choosing values for the noise variables $x_1,\dots, x_m$ over the field $\mathbb{F}_p$ and evaluating ciphertext integer values $\bar{\mathcal{P}}_1$ and $\bar{\mathcal{P}}_2$. That is, the ciphertext consists of two integer values $\mathcal{C} = (\bar{\mathcal{P}}_1, \bar{\mathcal{P}}_2),$ where 
\begin{align}\label{eq:cipher}
    &\bar{\mathcal{P}}_1 = \sum_{i=0}^{n}\sum_{j=1}^{m} \mathcal{P}_{ij1} (x_j x^i \ \text{mod} \ p)\\
    &\nonumber\bar{\mathcal{P}}_2 = \sum_{i=0}^{n}\sum_{j=1}^{m} \mathcal{P}_{ij2} (x_j x^i \ \text{mod} \ p). 
\end{align}
Here, $\mathcal{P}_{ij1}$ and $\mathcal{P}_{ij2}$ denote the cipher coefficients encrypted with the homomorphic encryption operators. Note that the cipher polynomials have coefficients in the hidden ring $\mathbb{Z}_S$, and all monomial calculations are performed $\text{mod} \ p$, the rest of the arithmetic is performed over integers. The values $\bar{\mathcal{P}}_1$, and $ \bar{\mathcal{P}}_2$ are integers forming the ciphertext $C=\{\mathcal{P}_1, \mathcal{P}_2\}$.  

\subsubsection{Decryption}
It is easy to verify that 
 the HPPK map as in Eq.\eqref{eq:p1} and Eq.\eqref{eq:p2},  under construction as shown in Eq.\eqref{eq:pbf}, holds a division invariant property on the multiplicand or the base multivariate polynomial $b(x, x_1, x_2, \dots,x_m)$. Indeed, 
 $$\frac{p_1(x, x_1, \dots,x_m)}{p_2(x, x_1, \dots,x_m)} = \frac{b(x, x_1, \dots,x_m)f_1(x)}{b(x, x_1, \dots,x_m)f_2(x)} = \frac{f_1(x)}{f_2(x)}.$$
The first step in the decryption process is to apply the functional homomorphic decryption operator to the ciphertext to recover plain polynomial values $\bar{p}_1$ and $\bar{p}_2$, which are evaluation results of plain multivariate polynomials $p_1(x, x_1, \dots,x_m)$ and $p_2(x, x_1, \dots,x_m)$ at the chosen message and noises respectively. This can be done as
 $$\hat{\mathcal{E}}_{(R_1^{-1}, S)}\mathcal{\bar{P}}_1 = p_1(x, x_1, \dots,x_m) = \bar{p}_1$$ $$\hat{\mathcal{E}}_{(R_2^{-1}, S)}{\mathcal{\bar{P}}}_2 = p_2(x, x_1, \dots,x_m) = \bar{p}_2.$$ 
 These values are used to compute the ratio $K$ modulo $p$ of the form 
\begin{equation}\label{eq:K}
   K = \frac{\bar{p}_1}{\bar{p}_2}
   =\frac{p_1(x, x_1, \dots,x_m)}{p_2(x, x_1, \dots,x_m)} 
   = \frac{f_1(x)}{f_2(x)} \ \text{mod} \ p
\end{equation}
  Note that the noise vector $\vec{x}_r$ is automatically eliminated through the division. The secret $x$  can then be found from Eq.\eqref{eq:K} by radicals if $f_1(x)$ and $f_2(x)$ are solvable such as linear or quadratic polynomials. 
  
  Note that when $\lambda > 1$, an extra 8-bit flag $\sigma$ should be added to the plaintext to distinguish the correct plaintext during the decryption procedure. We propose a formatted plaintext $X=(\sigma|x)$, with $\sigma$ to be a one byte flag. That is, $\sigma$ is concatenated with $x$, such that the most significant $8$ bits of $X$ are set as $\sigma$ and the remaining bits as $x$. The flag $\sigma$ can be generated by a cyclic redundancy check or CRC with the secret $x$. There are different CRC algorithms such as CRC-8, CRC-32, etc. 8 bits of CRC codes should be sufficient to make a right decision from the roots obtained during decryption. After decryption with successful flag verification, the shared secret would be established by removing the most significant $8$ bits of the obtained value $X$. Note that the field size should account for the flag $\sigma$. In this work, however, we focus mainly on the case $\lambda =1$. 

This division invariant property is the foundation for the HPPK encapsulation to be indistinguishable under chosen plaintext attacks.

\subsection{A Toy Example}
We demonstrate how HPPK works with a toy example. 
\subsubsection{Key Pair Generation}
Considering a prime field $\mathbb{F}_{13}$ with the prime $p=13$ and two noise variables $x_1, x_2$ for the simplicity of the demonstration purpose only, we can choose the hidden ring characterized by an integer of length $> 12$ bits. The private key consists of the following values:
\begin{itemize}
    \item [$\bullet$]$S = 6798, R_1 = 4267, R_2 = 6475$
    \item [$\bullet$] $f_1(x) = 4 + 9 x$
    \item [$\bullet$] $f_2(x) = 10 + 7 x$
    \item [$\bullet$] $B(x, x_1, x_2) = (8 + 7 x)x_1 + (5 + 11 x)x_2$ (note: just for key pair construction procedure; this polynomial is not stored in the memory)
\end{itemize}

The plain public key or PPK is simply constructed as
\begin{itemize}
    \item [$\bullet$] $P_1(x, x_1, x_2) = f_1(x)B(x, x_1, x_2) \ \text{mod} \ 13 = x_1(6+9x+11x^2) + x_2(7+11x+8x^2),$
    \item [$\bullet$] $P_2(x, x_1, x_2) = f_2(x)B(x, x_1, x_2) \ \text{mod} \ 13 = x_1(2+9x+10x^2) + x_2(11+2x+12x^2).$
\end{itemize}
The PPK polynomials are encrypted with the self-shared key $R_1, R_2$ over the ring $\mathbb{Z}_S$
\begin{itemize}
    \item [$\bullet$] $\mathcal{P}_1(x, x_1, x_2) = \mathcal{E}_{(R_1^{-1}, S)}P_1(x, x_1, x_2) = x_1(5208+4413x+6149x^2) + x_2(2677+6149x+146x^2)$
    \item [] $ \Longrightarrow \mathcal{P}_1 = 
        \begin{pmatrix}
            5208 & 2677 \\
            4413 & 6149 \\
            6149 & 146
        \end{pmatrix} $
    \item [$\bullet$] $\mathcal{P}_2(x, x_1, x_2) = \mathcal{E}_{(R_1^{-1}, S)}P_2(x, x_1, x_2) = x_1(6152+3891x+3568x^2) + x_2(3245+6152x+2922x^2)$
    \item [] $\Longrightarrow \mathcal{P}_2 =
        \begin{pmatrix}
            6152 & 3245 \\
            3891 & 6152 \\
            3568 & 2922
        \end{pmatrix} $
\end{itemize}
to create the so-called CPK $\mathcal{P}_1$ and $\mathcal{P}_2$.
\subsubsection{Encryption}
We randomly choose variables from $\mathbb{F}_{13}$: $x=8, x_1=3, x_2 = 6.$ We, then, pre-calculate values $$x_{11} = xx_1 \ \text{mod} \ 13 = 8 \times 3 \ \text{mod} \ 13 = 11,$$ $$x_{12} = x^2x_1 \ \text{mod} \ 13 = 8^2 \times 3 \ \text{mod} \ 13 = 10,$$ $$x_{21} = xx_2 \ \text{mod} \ 13 =8 \times 6 \ \text{mod} \ 13 = 9,$$ $$x_{22} = x^2x_2 \ \text{mod} \ 13 = 8^2 \times 6 \ \text{mod} \ 13 = 7.$$ Now we can calculate the ciphertext $C = \{198082, 192229\} $ as follows
\begin{itemize}
    \item [$\bullet$] $\bar{\mathcal{P}}_1 = x_1(5208+4413x+6149x^2) + x_2(2677+6149x+146x^2) = 198082$
    \item [$\bullet$] $\bar{\mathcal{P}}_2 = x_1(6152+3891x+3568x^2) + x_2(3245+6152x+2922x^2) = 192229$
\end{itemize}
\subsubsection{Decryption}
We first perform the homomorphic decryption to rebuild the plain polynomial equations
\begin{itemize}
    \item [$\bullet$] $P_1(x, x_1, x_2) =  f_1(x)B(x, x_1, x_2) = [\mathcal{E}_{(R_1^{-1}, S)} 19808] \ \text{mod}  \ 13 = [\frac{19808}{4267} \ \text{mod}  \ 6798] \ \text{mod}  \ 13 = 8$
    \item [$\bullet$] $P_2(x, x_1, x_2) =  f_2(x)B(x, x_1, x_2)=[\mathcal{E}_{(R_2^{-1}, S)} 192229] \ \text{mod} \ 13 = [\frac{192229}{6475} \ \text{mod} \ 6798] \ \text{mod} \ 13 = 9$
\end{itemize}
then we can eliminate the noise introduced by  the base multivariate polynomial
\begin{equation*}
     \frac{P_1(x, x_1, x_2)}{P_2(x, x_1, x_2)} = \frac{4 + 9 x}{10 + 7 x} = \frac{8}{9} \ mod \ 13 = 11
\end{equation*}
where the secret $x$ can be easily extracted as $x=8$. The encryption can be done with any possible values for $x_1$ and $x_2$ at a given secret $x$, which would produce different ciphertext $C$, but the decryption would reveal the same secret. This simple toy example demonstrates its capability of randomized encryption. 

\section{HPPK Security Analysis}\label{sec:security}

In this section, we analyze the security of the proposed HPPK algorithm. The security of HPPK relies on computational hardness of Modular Diophantine Equation, introduced in Definition~\ref{MDEPDef}, and Hilbert's tenth Problem, introduced in Definition~\ref{Hilbert10}. 
We begin by proving that HPPK satisfies the IND-CPA indistinguishability property. These results are then extended to prove that task of recovering plaintext from ciphertext in the framework of HPPK is NP-complete, and state its classical and quantum complexity. Afterwards, we focus on the private key attack and prove that the problem of obtaining the private key from the public key is NP-complete. Here we also provide classical and quantum complexities of obtaining privates key from public key. 
 
\subsection{Plaintext attack}

An attentive reader will notice that the evaluated ciphertext as illustrated in Eq.~\eqref{eq:cipher} has not been reduced modulo any integer. Thus, an adversary looking to perpetrate an attack to recover the plaintext can treat the coefficients of the polynomials in Eq.~\eqref{eq:cipher} and evaluated ciphertext as integers. The plaintext values, sought after by the adversary, are elements of the field $\mathbb{F}_p$, thus the malicious party can reduce the public values of the ciphertext modulo $p$ to solve for plaintext variables in the Eq~\eqref{eq:cipher}. We formally phrase it in the following remark. 

\begin{remark}\label{Modp}
For the purpose of obtaining the plaintext, the ciphertext and cipher coefficients as illustrated in the Eq.~\eqref{eq:cipher} can be considered modulo $p$ as follows
\begin{equation}\label{ciphermodp} \mathcal{C} = 
 \begin{cases}
    \sum_{i=0}^{n}\sum_{j=1}^{m} \mathcal{P}_{ij1} x_j x^i - \bar{\mathcal{P}}_1 = 0 \ \text{(mod} \ p)\\
    \sum_{i=0}^{n}\sum_{j=1}^{m} \mathcal{P}_{ij2} x_j x^i - \bar{\mathcal{P}}_2 = 0\ \text{(mod} \ p) \\
    \end{cases}  
\end{equation}
\end{remark}

\begin{definition}[Modular Diophantine Equation]\label{MDEPDef}
The Modular Diophantine Equation asks whether an integer solution exists to the equation 
$$
P(y_1,\ldots,y_k) - 1 = 0 \mod p,
$$
given as an input of a polynomial $P(y_1, \dots, y_k)$ and a prime $p.$
\end{definition}
\begin{remark}
A positive answer to this question would include a solution. 
\end{remark}

Let $m+1>2$. Note that the system in the Eq.~\eqref{ciphermodp} can be normalized as 
\begin{equation}\label{ciphermodpnorm}\mathcal{C} = 
\begin{cases}
 \sum_{i=0}^{n}\sum_{j=1}^{m} \mathcal{P}'_{ij1} x_j x^i - 1 =0 \ \text{(mod} \ p), \\ 
 \sum_{i=0}^{n}\sum_{j=1}^{m} \mathcal{P}'_{ij2} x_j x^i - 1 =0 \ \text{(mod} \ p).
 \end{cases}
\end{equation} The most naive way of solving such normalized system is to solve each equation and find a common solution. Each such equation is an instance of a Modular Diophantine Equation. The more obvious way to solve the system in Eq.~\eqref{ciphermodp} would be to use Gaussian elimination and transform the system to a single equation. Indeed, since the coefficients of the ciphertext are publicly known, and the noise variables are linear in the ciphertext, the adversary can express any noise variable using the remaining terms of the equation and reduce the system to a single equation of the form
\begin{equation}\label{singlecipher}
H(x, x_1, \dots, x_{m-1})-1 = 0    
\end{equation}over $\mathbb{F}_p$ with $m$ unknowns, where $m>1$. We assume that the adversary favours the ciphertext form with less variables. Thus, from the perspective of the adversary the cipheretext has form as in Eq.~\eqref{singlecipher}. From here on forward, we consider the ciphertext in the form given in Eq.~\eqref{singlecipher}. We formally define said form below.

\begin{definition}
Let $m+1>2$. The ciphertext in its normalized reduced form is a single equation 
\begin{equation}\label{reducednormalform}
H(x, x_1, \dots, x_{m-1})-1 = 0     
\end{equation} over $\mathbb{F}_p$, where $x$ corresponds to the plaintext variable and the remaining variables are noise variables. 
\end{definition}

Note that even in its normalized reduced form the ciphertext is an instance of a Modular Diophantine Equation. Since $m+1>2$ we can argue that the adversary does not benefit much by reducing the system in Eq.~\eqref{ciphermodpnorm} to a single equation~\eqref{reducednormalform}, and eliminating one variable. The number of expected solutions to the Eq.~\eqref{reducednormalform} remains $p^{m-1}$, and the adversary is facing with the problem of deciding which solution is the correct one. That is, a brute-force search algorithm can find a list of solutions to the Eq.~\eqref{reducednormalform} by trying all the possible $m-1$ variables values over $\mathbb{F}_p$. The adversary is interested in a particular solution from the list. 

One might argue that the attacker is interested only in the plaintext variable $x \in \mathbb{F}_p$. Thus, the adversary can simply guess the value $x$. The complexity of this guess is $\mathcal{O}(p)$. However, note that the guess has to be tested for correctness. This will require coming up with noise variables and testing whether the guess is correct. Moreover, NIST requires the size of the actual communicated secret to be $32$ bytes. Thus, the secret that is transferred between two parties consists of $K$ blocks, where each block is $p$ bits. Each block corresponds to the HPPK secret $x$. The secret message is then $K$ different values $x$ concatenated together to form a $32$ byte secret. Each such block $x$ is encrypted separately using HPPK. The complexity of correctly guessing the transferred secret message is then $O(p^4)$.

\begin{theorem}\label{theo:MDEPisNPcomplete}
The Modular Diophantine Equation Problem is NP-complete.
\end{theorem}
 \begin{proof}
The proof, using the Boolean Satisfiability Problem, is given by Moore and Meterns~\cite[Section 5.4.4]{moore2011nature}. 
 \end{proof}
 
Theorem~\ref{theo:MDEPisNPcomplete} states that a brute-force search algorithm can find a solution to the Modular Diophantine Equation by trying all the possible solutions. Thus, without loss of generality, we treat the ciphertext-only attack on a ciphertext in its normal reduced form as a Modular Diophantine Problem. Indeed, by Theorem~\ref{theo:MDEPisNPcomplete} the algorithm to find a solution to a Modular Diophantine Equation does not simply terminate to give a solution, it is a brute-force search algorithm that considers every possible solution before producing a result. In other words, it goes through all the possibilities to choose the correct one.

\subsubsection{IND-CPA Indistinguishability Property and Ciphertext only Attack}
\label{sec:MPPKhardness}

We suppose that the adversary will choose to perpetrate the attack on the ciphertext in its normal reduced form as in Eq.~\eqref{reducednormalform}, for its easier to attack. In the framework of HPPK, the public key elements are the coefficients of the ciphertext polynomials. Thus, if the ciphertext is presented in its reduced normalized form, as defined in Eq.~\eqref{reducednormalform}, setting the coefficients of such polynomial to be the ciphertext does not disadvantage the adversary. 

 \begin{theorem}[HPPK has IND-CPA property]
Let $m>1$, where $m$ is the total number of variables in the normalized reduced form of the ciphertext as in Eq.~\eqref{reducednormalform}. If the Modular Diophantine Equation is NP-complete,
the HPPK encryption system is provably secure in the IND-CPA security model with a reduction loss of $p^{m-2}$.
 \end{theorem}
 \begin{proof}
Assume that there exists an adversary $\mathcal{A}$ that $(t,\epsilon)$-breaks the HPPK encryption system in the IND-CPA security model. 
We construct a simulator $\mathcal{B}$ that solves the Modular Diophantine Equation.
 Given as input, a Modular Diophantine Equation instance  
$
   (p, H(x, x_1, \dots, x_{m-1}))
$, where $H(x, x_1, \dots, x_{m-1})$ is of the form~\eqref{reducednormalform} and $m>1$, the simulator $\mathcal{B}$ runs $\mathcal{A}$ as follows.
\sloppy
The simulator sets the normalized public key over $\mathbb{F}_p$ to the coefficients of the polynomial $H(x, x_1, \dots, x_{m-1})$. 
 \sloppy
 The challenge consists of the following game. 
 The adversary $\mathcal{A}$ generates two distinct messages $m_0$ and $m_1 \in \mathbb{F}_p$, and submits them to the simulator.
 The simulator $\mathcal{B}$ randomly chooses $b$ in $\{ 0, 1 \}$ as well as random values $r_1, \dots, r_{m-1}$ for the noise variables, and sets the ciphertext to be the value $$\bar{H} = H(m_b, r_1, \dots, r_{m-1}).$$ The challenge for the adversary then consists of the following equation to be solved for $x$:
 \sloppy
 $$\hat{H}(x, x_1, \dots, x_{m-1})-1 = 0$$ over $\mathbb{F}_p$. Here, $$\hat{H}(x, x_1, \dots, x_{m-1}) = \frac{1}{\bar{H}}H(x, x_1, \dots, x_{m-1}).$$ The challenge remains to be the Modular Diophantine Equation $H(x, x_1, \dots, x_{m-1})-1 =0$, since the value $\frac{1}{\bar{H}}$ can be pushed to the noise variables, which are random and do not influence the plaintext. Indeed, let $h_{ij}$ be the coefficients of the polynomial $H(x, x_1, \dots, x_{m-1})$ for any $i \in \{0, \dots, n\}$ and $j \in \{1, \dots, {m-1}\}$, then $$\sum_{i=0}^{l}\sum_{j=1}^{m-1}h_{ij}x^ix_j\times \frac{1}{H(m_b, r_1, \dots, r_{m-1})} = \sum_{i=0}^{l}\sum_{j=1}^{m-1}h_{ij}x^ix'_j,$$ where $x'_j = \frac{x_j}{\bar{H}}.$ The challenge in this case is correct, as it corresponds to the challenged plaintext and remains in the form of a Diophantine equation chosen by the simulator. 
 
 	
The coefficients of the challenge equation come from the submitted Diophantine equation, and thus, from the point of view of the adversary are random. The values $r_1, \dots, r_{m-1}$ are selected at random. The adversary does not have knowledge of the values $\{x'_1, \dots, x'_{m-1}\}$ and they can not be calculated from the other parameters given to the adversary. So the noise variables $x_j'$ for all $j\in\{1, \dots, m-1\}$ are random. 
Hence, the simulation holds randomness property. By construction, the simulation is indistinguishable from a real attack. 
That is, the adversary is challenged with solving the equation as in the Eq.~\eqref{reducednormalform}, which is HPPK ciphertext in its normalized reduced form.  

There is no abort in the simulation. The adversary outputs a random guess $b'$ of $b$. When $b'$ is equal to $b$, the adversary wins. Otherwise, the adversary looses. The probability of simply guessing the value for $x$ is $Pr = \frac{1}{2}$. We will calculate the probability of solving the IND-CPA challenge with the advantage of the adversary, that is $Pr = \frac{1}{2} + \alpha$. The advantage comes from the assumption that the adversary can break the HPPK cryptosystem.
 	
The challenge has a general form as in the Eq.~\eqref{reducednormalform}, thus, the equation is expected to have $p^{m-1}$ distinct solutions, considering all $m$ variables. On the other hand, it is known that the variable $x \in \{m_0, m_1\}$.  Assuming $x = m_0$, there are now $p^{m-2}$ possible solutions to choose the correct solution from. The same is true for $x=m_1$. That is, the probability of finding correct solution of the equation $H(x, x_1, \dots, x_{m-1}) - 1 =0$ is 
$$Pr(\text{correct solution} | x_0 = m_0) =  Pr(\text{correct solution} | x = m_1) = \frac{1}{p^{m-2}},$$ 
where $Pr(\text{correct solution})$ denotes probability of finding the correct solution to the equation $H(x, x_1, \dots, x_{m-1}) - 1 =0$. Then by the law of total probability, the probability of solving the challenge equation is 
$$Pr(\text{correct solution}) = Pr(\text{correct solution} | x = m_0)Pr(x = m_0) + Pr(\text{correct solution} | x = m_1)Pr(x = m_1) = \frac{1}{p^{m-2}}.$$ 
Accounting for the advantage that the adversary has, the probability $\alpha$ is $Pr(\text{correct solution}) = \frac{\epsilon}{p^{m-2}}$. The total probability of solving the IND-CPA challenge is then $\frac{1}{2} + \frac{\epsilon}{p^{m-2}}.$ 
 
The simulation is indistinguishable from a real attack. So the adversary who can break the challenge ciphertext will uncover the solution to the given Modular Diophantine Equation problem. The probability of breaking the ciphertext is $\frac{\epsilon}{p^{m-2}}$. 
 
The advantage of solving the Diophantive Equation problem is then
 $\frac{\epsilon}{p^{m-2}}.$
 Let $T_s$ denote the time cost of the simulation.
 We have $T_s = \mathcal{O}(1)$.
 The simulator $\mathcal{B}$ solves the Modular Diophantine Equation with time cost and advantage $(t+T_s,\epsilon/{p^{m-2}}) = (t,\epsilon/{p^{m-2}}) $. Thus, contradicting the Theorem~\ref{theo:MDEPisNPcomplete} so the initial assumption is wrong.
 \end{proof}

The framework of the IND-CPA challenge entails known plaintext, in other words, the adversary knows that the secret $x \in \{m_0, m_1\}.$ We now state the complexity of the unknown plaintext ciphertext-only attack.

\begin{lemma}[Ciphertext-only attack]\label{complexplaintext}
Let $m+1>2$. The classical complexity of finding the plaintext from the ciphertext is $O(p^{m-1}).$
\end{lemma}
\begin{proof}
Let the adversary favour the ciphertext in its reduced normal form~\eqref{reducednormalform}. Without any knowledge about the plaintext, the adversary will need to solve the Eq.~\eqref{reducednormalform} to obtain the plaintext along with the noise variables. A single equation over $\mathbb{F}_p$ with $m$ variables is expected to have $p^{m-1}$ possible solutions over $\mathbb{F}_p$. The correct one is among them. That is, the plaintext encapsulated in a single variable $x$ is not the sole variable in the ciphertext equation. However, it is the only unknown of interest. The adversary can try and simply guess $x$, the complexity of the guess is $\mathcal{O}(p)$. However, they have to test whether their guess is correct. Moreover, the secret transferred between the communicating parties consist of $32$ bytes as required by NIST. Thus, the adversary will have to guess $K$ many values for $x$, where $K = \frac{32\times 8}{p}$. In this case, the complexity is $\mathcal{O}(p^{K}).$ We expect $K > m-1.$
Quantum complexity of the described attack due to Grover's search algorithm is $O(p^{\frac{m-1}{2}}).$ 
\end{proof}

\subsection{Private key attack}

\begin{lemma}\label{plainattack}
Let $\lambda \le 2$. There exists a polynomial time algorithm to find coefficients of univariate polynomials $f_1(x_0)$ and $f_2(x_0)$ given the plain central maps $P_1$ and $P_2$.
\end{lemma} 
\begin{proof}
Note that all the plain coefficients of the polynomials $p_1(x, x_1, \dots, x_m)$ and $p_2(x, x_1, \dots, x_m)$ as defined in Eq.~\eqref{eq:bff} are defined over the prime field $\mathbb{F}_p$. Thus, for any fixed $j$, it is possible to use Gaussian elimination to reduce the system of equations formed by the plain coefficients of $p_1(x, x_1, \dots, x_m)$ and $p_2(x, x_1, \dots, x_m)$ of the form
\begin{equation}
    \begin{cases}
        f_{z0}b_{0j} = p_{0jz}, \\
        f_{z1}b_{0j} + f_{z0}b_{1j} = p_{1jz},\\
        \vdots\\
        f_{z\lambda}b_{n_bj} = p_{njz},
    \end{cases}
\end{equation} 
where $z \in \{1,2\}$ corresponding to either plain polynomial $p_1(x, x_1, \dots, x_m)$ or $p_2(x, x_1, \dots, x_m)$ for any given noise variable $x_j$. Gaussian elimination would produce a single polynomial in $\lambda$ variables, namely $\frac{f_{z2}}{f_{z0}}$ and $\frac{f_{z1}}{f_{z0}}$ for $\lambda = 2$ or $\frac{f_{z1}}{f_{z0}}$ if $\lambda = 1$. Such univariate or bivariate equation is solvable. Depending on the HPPK parameters, the adversary can simply use radical solutions, Evdokimov's algorithm~\cite{evdokimov1994factorization}, or resultants together with Evdokimov's algorithm to solve such equation~\cite{evdokimov1994factorization,Stiller2004AnIT}. Gaussian elimination can be performed in polynomial time, and finding solutions by radicals, Evdokimov's algorithm and computing resultants all have polynomial time complexity~\cite{evdokimov1994factorization,Stiller2004AnIT}. 
\end{proof}

\begin{lemma}~\label{threevalues}
Let $\lambda \le 2$. Finding private key from the cipher public key in the framework of HPPK reduces to finding the homomorphic encryption key $S, R_1,$ and $R_2.$
\end{lemma}
\begin{proof}
The private key consists of the coefficients of the univariate polynomials $f_1(x), f_2(x)$ as well as values $S, R_1, R_2$ used to encrypt the plain public key to the cipher public key.  By Lemma~\ref{plainattack} once the values $R_1, R_2$ and $S$ are known, the coefficients of $f_1(x), f_2(x)$ can be found in polynomial time. 
\end{proof}

\begin{definition}[Diophantine set]\label{Diophset}
The Diophantine set is a set $S \subset \mathbb{N}$ associated with a Diophantine equation $P(b, a_1, \dots, a_k) \in \mathbb{Z}[b, a_1, \dots, a_k],$ where $k>0$ such that $$b \in S \text{ if and only if } (\exists a_1, \dots, a_k)(P(b, a_1, \dots, a_m) = 0)$$
\end{definition}

\begin{theorem}[MRDP Theorem]\label{MRDP}
The Matiyasevich–Robinson–Davis–Putnam (MRDP) theorem states that every computably enumerable set is Diophantine, and every Diophantine set is computably enumerable.
\end{theorem}
\begin{proof}
The result has been proven in various works, for instance~\cite{MRDP}. 
\end{proof}

\begin{theorem}[Hilbert's tenth problem]\label{Hilbert10}
Hilbert's tenth problem asks whether the general Diophantine Problem is solvable. Due to MRDP, Hilbert's tenth problem is undecidable.
\end{theorem}
\begin{proof}
For proof see~\cite{MRDP}.
\end{proof}

\begin{theorem}\label{complexprivkey}
Private key attack is non-deterministic and has complexity of at least $O(T^{3})$, where $T$ is the largest number with $2|p|_2 + |L|_2$ bit-length.  
\end{theorem}
\begin{proof}
By Lemma~\ref{plainattack} and~\ref{threevalues} the attack on public key reduces to finding the values $S, R_1$, and $R_2.$ From the perspective of the attacker, the values $S, R_1$, and $R_2$ could be treated as a one-time pad keys as they have been chosen at random, and can not be calculated from other parameters given to the attacker. An obvious attack would be a brute force search for all the three values, $S, R_1$, and $R_2.$ The direct brute force search classical complexity would be greater than $O(T^{3})$ for the three values together, where $T$ is the largest $(2|p|_2 + |L|_2)$-bit number. Due to Grover's algorithm, the quantum complexity is greater than $O(T^{\frac{3}{2}}).$ Note however, because of the condition $gcd(S, R_1) = gcd(S, R_2) = 1$, once $S$ is found the search span for $R_1$ and $R_2$ reduces. Brute force search entails a non-deterministic result, however, we provide a more formal argument below. 

For each fixed chose of $j$, each public key coefficient can be written in the integer domain as follows
\begin{equation} \label{eq:bb}
    \begin{cases}
        f_{z0}b'_{z0j} = r_{0jz}S + \mathcal{P}_{0jz}, \\
        f_{z1}b'_{z0j} + f_{z0}b'_{z1j} =r_{1jz}S + \mathcal{P}_{1jz},\\
        \vdots\\
        f_{z\lambda}b'_{zn_bj} = r_{njz}S + \mathcal{P}_{njz},
    \end{cases}
\end{equation} 
with $j = 1, \dots, m$, $z=1, 2$, and $b'_{zij} = R_{z}b_{ij}$. Here, $R_{z}$, and $S$ are unknowns from the hidden ring $\mathbb{Z}_S$. Values $r_{ijz}$ are merely some unknown integers. The only known values are those of the form $\mathcal{P}_{ijz}$. Using Gaussian eliminations, all unknowns of the form $b'_{zij}$  can be eliminated, and the equation system in Eq~\eqref{eq:bb} can be reduced to a single equation over $\mathbb{Z}$
\begin{equation}
    P(f_{z0},\dots,f_{z\lambda}, r_{0jz}, \dots, r_{njz}, S) - \bar{P} = 0.
\end{equation}
Solving such equation by Theorem~\ref{MRDP} and~\ref{Hilbert10} is an NP-complete task. For each $j$, we can generate one such equation. Considering them all together, the adversary will arrive at an underdetermined system as the variables in the system depend on $j$. Each equation in such a system is a multivariate Diophantine equation. One way to solve this system is to solve each equation separately and search for common solutions. However, by Theorem~\ref{MRDP} and~\ref{Hilbert10} this is an NP-complete problem. Reducing the system to a single polynomial still produces a multivariate Diophantine equation, solving which is an NP-complete problem by Theorem~\ref{MRDP} and~\ref{Hilbert10}.  
\end{proof}

\subsection{Security Conclusion}

At large, the security of the HPPK cryptosystem relies on the problem of solving undetermined system of equations over $\mathbb{F}_p$. Such system is expected to have $p^{n-m}$ possible solutions, where $n$ is the number of variables and $m$ is the number of equations in the system. The attacker can solve this system to find all possible solutions, however, it is the problem of determining the correct solution from all the possible solutions that makes HPPK secure. 

The ciphertext attack requires the adversary to solve an underdetermined system of equations over $\mathbb{F}_p$, which can be reduced to a single Modular Diophantine equation. Solving this equation is an NP-complete problem. 

The public key attack aimed to unveil the plaintext reduces to a brute force search for three unknown values $S, R_1, R_2.$ To find these values, the attacker can either use brute-force search or solve an underdetermined system of equations over the integers. The former yields non-deterministic results and the latter is an NP-complete problem.   

We conclude that from the point of view of the adversary, the following is true. 
\begin{proposition}
The best classical complexity to attack HPPK is $O(p^{m-1})$.
\end{proposition}
\begin{proof}
We assume that the malicious party will take the most advantageous path for them. Thus, by Lemma~\ref{complexplaintext}, Lemma~\ref{plainattack}, Lemma~\ref{threevalues}, and Theorem~\ref{complexprivkey} we can conclude that the best attack is to obtain the plaintext from the ciphertext. Such attack is non-deterministic with classical complexity of $O(p^{m-1}).$
\end{proof}





\section{Brief Benchmarking Performance}\label{sec:bench}
To account for the best complexity of $O(p^{m-1})$, we recommend the following configuration to achieve NIST security levels I, III, and V, as illustrated in Table~\ref{tab:Config}. 

\begin{table}[htbp]
\caption{Configuration of HPPK for different NIST Security Levels.}
\begin{center}
\begin{tabular}{|c|c|c|c|}
\hline
&\multicolumn{3}{|c|}{\textbf{Security Level}} \\
\cline{2-4} 
 \textbf{Configuration} & \textbf{\textit{Level I}}& \textbf{\textit{Level III}}& \textbf{\textit{Level V}}\\
\hline
($\log p, n_b, \lambda, m$)& (64, 1, 1, 3) & (64, 1, 1, 4), & (64, 1, 1, 5)\\
($\log p, n_b, \lambda, m$)& (64, 2, 1, 3) & (64, 2, 1, 4), & (64, 2, 1, 5)\\
\hline
\end{tabular}
\label{tab:Config}
\end{center}
\end{table}

To measure the performance of the HPPK algorithm we used benchmarking toolkit, called the SUPERCOP. NIST PQC finalists used the SUPERCOP for their benchmarking and have contributed the results to the platform, thus, we take advantage of the available resources, and report on the performance of HPPK alongside with the NIST PQC schemes, namely, McEliece, Kyber, NTRU, and Saber algorithms. From now on we refer to them as the NIST finalists. For our work we used a 16-core Intel\textregistered Core\texttrademark i7-10700 CPU at 2.90 GHz system. We have not, however, configured the AVX solution for optimized HPPK performance. Therefore, comparisons in this benchmarking performance are set to the reference mode for all finalists but in the same computing system.

We start by illustrating the parameter set of all the measured primitives for all three security levels in the Table~\ref{tab:Parameters}. As required by NIST, the secret is set at $32$ bytes. The data illustrates that for each security level, HPPK offers considerably small public key sizes and ciphertext sizes for all three security levels compared to all NIST finalists, except for the ciphertext size of McEliece at level I and level III. We point out that, HPPK offers the same secret key size of 83 bytes and ciphertext size of 208 bytes for all three levels. In comparison with the NIST standardized KEM algorithm Kyber, HPPK's public key sizes are less than half of respective public key sizes for Kyber for all three security levels. Secret key sizes for Kyber are about 20 times bigger at level I and close to 40 times bigger at level V than those of HPPK. As for the ciphertext sizes, Kyber demonstrates 3.7-7.5 times bigger ciphertext sizes than those of HPPK.

\begin{table}[htbp]
\caption{Parameter set of the measured primitives for NIST security levels I, III, and V, given the secret size of $32$ bytes}
\begin{center}
\begin{tabular}{|c|c|c|c|c|c|c|c|c|c|}
\hline
\textbf{Crypto}&\multicolumn{9}{|c|}{\textbf{Size (Bytes)}} \\
\cline{2-10} 
\textbf{system}& \multicolumn{3}{|c|}{\textbf{Level I}} & \multicolumn{3}{|c|}{\textbf{Level III}}&\multicolumn{3}{|c|}{\textbf{Level V}}\\
\cline{2-10}
\textbf{} & \textbf{\textit{PK}}$^{\mathrm{1}}$& \textbf{\textit{SK}}$^{\mathrm{1}}$& \textbf{\textit{CT}}$^{\mathrm{1}}$& \textbf{\textit{PK}}& \textbf{\textit{SK}}& \textbf{\textit{CT}}&\textbf{\textit{PK}}& \textbf{\textit{SK}}& \textbf{\textit{CT}}\\
\hline
McEliece$^{\mathrm{2}}$~\cite{McEliece1978}&  261120 & 6492 & 128 & 524,160 & 13,608 & 188 & 1,044,992 & 13,932 & 240\\
NTRU$^{\mathrm{4}}$~\cite{Hoffstein1998}& 699 & 935 & 699 & 930 & 1,234 & 930 & 1,230 & 1,590 & 1,230\\
Saber$^{\mathrm{5}}$~\cite{vercauterensaber} & 672	& 1568	& 736 & 1,312	& 3,040	& 1,472 & 1,312	& 3,040	& 1,472\\
Kyber$^{\mathrm{3}}$~\cite{avanzi2017crystals} &800	& 1632 & 768 & 1,184	& 2,400 & 1,088 & 1,568	& 3,168 & 1,568\\
HPPK($n_b=1$)$^{\mathrm{6}}$ & 306 & 83 & 208 & 408 & 83 & 208 & 510 & 83 & 208\\
HPPK($n_b=2$)$^{\mathrm{6}}$ & 408 & 83 & 208 & 544 & 83 & 208 & 680 & 83 & 208\\
\hline
\end{tabular}
\end{center}

\footnotesize{$^{\mathrm{1}}$ We denote the secret key as \textbf{\textit{SK}}, the public key as \textbf{\textit{PK}}, and the ciphertext as \textbf{\textit{CT}}}.\\
\footnotesize{$^{\mathrm{2}}$ \textit{mceliece348864} primitive was measured for Level I, \textit{mceliece460896} primitive was measured for Level III, and \textit{mceliece6688128} for Level V}\\
\footnotesize{$^{\mathrm{3}}$ \textit{Kyber512} primitive was measured for Level I, \textit{Kyber768} primitive was measured for Level III, and \textit{Kyber1024} for Level V}\\
\footnotesize{$^{\mathrm{4}}$ \textit{NTRUhps2048509} primitive was measured for Level I, \textit{ntruhps2048677} primitive was measured for Level III, and \textit{ntruhps4096821} for Level V}\\
\footnotesize{$^{\mathrm{5}}$ \textit{Light Saber} primitive was measured for Level I, \textit{Saber} primitive was measured for Level III, and \textit{FireSaber} for Level V}\\
\footnotesize{$^{\mathrm{6}}$ For each security level, HPPK primitive is configured as shown in Table~\ref{tab:Config}}

\label{tab:Parameters}
\end{table}

Table~\ref{tab:Keygen} provides the reader with median values in clock cycles of the measurement results for the key generation procedure of all the primitives configured to provide security levels I, III, and V. To provide a bigger picture we include results for RSA-2048. The results correspond only to security level I, as RSA-2048 provides $112$ bits of entropy. The reader can see that HPPK key generation performance is rather fast, with median clock cycles of over $18,000$ for $n_b=1$ and $22,000$ for $n_b=2$  for level I, over $21,000$ for $n_b=1$ and $28,000$ for $n_b=2$ for level III, and over $26,000$ for $n_b=1$ and $34,000$ for $n_b=2$ for level V. The fastest key generation performance among the NIST finalists is offered by Saber and Kyber, with median values of over $39,000$ and $72,000$ clock cycles respectively for level I, median values over $115000$ for level III, and median values of $128,000$ clock cycles for level V. Compared to the standard algorithm Kyber, HPPK demonstrates a 3.5-6 times faster key generation performance. The remaining primitives measured, including RSA, display median values of over $6$ million clock cycles for level I, over $10$ million for level III, and over $16$ million clock cycles for level V.

\begin{table}[htbp]
\caption{Median values of the key generation performance for NIST security levels I, III, and V.}
\begin{center}
\begin{tabular}{|c|c|c|c|}
\hline
\textbf{Crypto}&\multicolumn{3}{|c|}{\textbf{Performance (Clock cycles)}} \\
\cline{2-4} 
\textbf{system} & \textbf{Level I}&  \textbf{Level III} &  \textbf{Level V}\\
\hline
McEliece$^{\mathrm{1}}$~\cite{McEliece1978}  & 152,424,455&	509,364,485	&	1,127,581,201\\
NTRU$^{\mathrm{3}}$~\cite{Hoffstein1998} &	6,554,031&	10,860,295&		16,046,953	\\
Saber$^{\mathrm{4}}$~\cite{vercauterensaber} & 	39,654&	128,935	&128,412\\
Kyber$^{\mathrm{2}}$~\cite{avanzi2017crystals} & 	72,403&		115,654	&	177,818	\\
HPPK ($n_b=1$)$^{\mathrm{5}}$ & 	18,034&	21,946&	26,603\\	
HPPK($n_b=2$)$^{\mathrm{5}}$ & 	22,625&	28,360&	34,719\\	
RSA-2048 & 91,985,129 &- &-  \\
\hline
\end{tabular}
\end{center}
\footnotesize{$^{\mathrm{1}}$ \textit{mceliece348864} primitive was measured for Level I, \textit{mceliece460896} primitive was measured for Level III, and \textit{mceliece6688128} for Level V} \\
\footnotesize{$^{\mathrm{2}}$ \textit{Kyber512} primitive was measured for Level I, \textit{Kyber768} primitive was measured for Level III, and \textit{Kyber1024} for Level V}\\
\footnotesize{$^{\mathrm{3}}$ \textit{NTRUhps2048509} primitive was measured for Level I, \textit{ntruhps2048677} primitive was measured for Level III, and \textit{ntruhps4096821} for Level V}\\
\footnotesize{$^{\mathrm{4}}$ \textit{Light Saber} primitive was measured for Level I, \textit{Saber} primitive was measured for Level III, and \textit{FireSaber} for Level V}\\
\footnotesize{$^{\mathrm{5}}$ For each security level, HPPK primitive is configured as shown in Table~\ref{tab:Config}}
\label{tab:Keygen}
\end{table}

We provide Table~\ref{tab:Enc} to illustrate encryption procedure performance of HPPK, NIST finalists, and RSA-2048. The table illustrates median values given in clock cycles. HPPK offers fast encryption with clock cycles from 17,000 for security level I to 25,000 for security level V, outperforming all mentioned NIST finalists. More specifically comparing the NIST standard algorithm, Kyber, to HPPK the table illustartes that Kyber offers 4-8 times slower encryption than HPPK for all levels. However, RSA-2048 offers the fastest encryption performance among all the other primitives measured for level I, due to its small public encryption key, usually chosen to be 65535. 

In Table~\ref{tab:Dec} we illustrate median values given in clock cycles for the decryption procedure corresponding to the HPPK algortihm, the NIST finalists algorithms, and RSA-2048. Table~\ref{tab:Dec} shows that HPPK offers fast decryption performance with median values for level I being at $28,000$ clock cycles. Meanwhile, median values for the faster NIST finalists are over $63,000$ and $117,000$ clcok cycles for Saber and Kyber respectively configured to provide level I security. Both RSA-2048 and NTRU have median values over $1$ million clock cycles for level I. McEliece median value is over $45$ million clock cycles. For levels III and V the account is similar. HPPK offers median values for both of these levels which fall in the interval of $[28000, 30000]$ clock cycles, while median values for Saber are over $170,000$ for levels III and V. Median values for Kyber fall into the interval of $[166000,238000]$ for levels III and V. NTRU displays median values of over $2$ million clock cycles for level III and V. McEliece offers the slowest decryption procedure with median values being over $93$ million for level III, and $179$ million for level V. HPPK decryption demonstrates a stable performance at $30,000$ clock cycles for all security levels, due to its special decryption mechanism with a modular division.

\begin{table}
\caption{Median values of the key encapsulation performance estimated in clock cycles for NIST security levels I, III, and V.}
\begin{center}
\begin{tabular}{|c|c|c|c|c|}
\hline
\textbf{Crypto}&\multicolumn{3}{|c|}{\textbf{Performance (Clock cycles)}} \\
\cline{2-4} 
\textbf{system} & \textbf{Level I}& \textbf{Level III}& \textbf{Level V} \\
\hline
McEliece$^{\mathrm{1}}$~\cite{McEliece1978} & 108,741& 172,538& 263,169\\
NTRU$^{\mathrm{3}}$~\cite{Hoffstein1998} &	418,622	&703,046&1,063,124\\
Saber$^{\mathrm{4}}$~\cite{vercauterensaber} &	62,154	&157,704	& 157,521\\
Kyber$^{\mathrm{2}}$~\cite{avanzi2017crystals} &	95,466&	140,376& 205,505\\
HPPK($n_b=1$)$^{\mathrm{5}}$ & 17,354 & 20,951&25,087\\
HPPK($n_b=2$)$^{\mathrm{5}}$ & 23,073 & 28,642&35,173\\
RSA-2048&13,429& - & -\\
\hline
\end{tabular}
\end{center}
\footnotesize{$^{\mathrm{1}}$ \textit{mceliece348864} primitive was measured for Level I, \textit{mceliece460896} primitive was measured for Level III, and \textit{mceliece6688128} for Level V}  \\
\footnotesize{$^{\mathrm{2}}$ \textit{Kyber512} primitive was measured for Level I, \textit{Kyber768} primitive was measured for Level III, and \textit{Kyber1024} for Level V} \\
\footnotesize{$^{\mathrm{3}}$ \textit{NTRUhps2048509} primitive was measured for Level I, \textit{ntruhps2048677} primitive was measured for Level III, and \textit{ntruhps4096821} for Level V} \\
\footnotesize{$^{\mathrm{4}}$ \textit{Light Saber} primitive was measured for Level I, \textit{Saber} primitive was measured for Level III, and \textit{FireSaber} for Level V} \\
\footnotesize{$^{\mathrm{5}}$ For each security level, HPPK primitive is configured as shown in Table~\ref{tab:Config}}
\label{tab:Enc}
\end{table}

\begin{table}
\caption{Median values of key decapsulation performance for NIST security levels I, III, and V.}
\begin{center}
\begin{tabular}{|c|c|c|c|c|}
\hline
\textbf{Crypto}&\multicolumn{3}{|c|}{\textbf{Performance (Clock cycles)}} \\
\cline{2-4} 
\textbf{system} & \textbf{Level I}& \textbf{Level III}& \textbf{Level V} \\
\hline
McEliece$^{\mathrm{1}}$~\cite{McEliece1978} & 45,119,775&	93,121,708&	179,917,369\\
NTRU$^{\mathrm{3}}$~\cite{Hoffstein1998} & 1,245,062&	2,099,254&	3,129,150\\
Saber$^{\mathrm{4}}$~\cite{vercauterensaber} &63,048&	173,712&	177,109\\
Kyber$^{\mathrm{2}}$~\cite{avanzi2017crystals} & 117,245	&166,062&	237,484\\
HPPK($n_b=1$)$^{\mathrm{5}}$ & 28,301 &28,759	& 	29,671\\
HPPK ($n_b=2$)$^{\mathrm{5}}$ & 29,791 &29,266	& 	29,364\\
RSA-2048& 1,670,173 & - & - \\
\hline
\end{tabular}
\end{center}
\footnotesize{$^{\mathrm{1}}$ \textit{mceliece348864} primitive was measured for Level I, \textit{mceliece460896} primitive was measured for Level III, and \textit{mceliece6688128} for Level V}  \\
\footnotesize{$^{\mathrm{2}}$ \textit{Kyber512} primitive was measured for Level I, \textit{Kyber768} primitive was measured for Level III, and \textit{Kyber1024} for Level V}  \\
\footnotesize{$^{\mathrm{3}}$ \textit{NTRUhps2048509} primitive was measured for Level I, \textit{ntruhps2048677} primitive was measured for Level III, and \textit{ntruhps4096821} for Level V}  \\
\footnotesize{$^{\mathrm{4}}$ \textit{Light Saber} primitive was measured for Level I, \textit{Saber} primitive was measured for Level III, and \textit{FireSaber} for Level V}  \\
\footnotesize{$^{\mathrm{5}}$ For each security level, HPPK primitive is configured as shown in Table~\ref{tab:Config}}
\label{tab:Dec}
\end{table}

\section{Conclusion}\label{conclusion}
In this paper, we introduced a new Functional Homomorphic Encryption, which in contrast with conventional homomorphic encryption, is intended to secure public keys of multivariate asymmetric cryptosystems. 
Functional homomorphic encryption is applied to polynomials, to leverage homomorphic properties and allow for user input through variables. The functional homomorphic encryption and decryption operators are multiplication operators modulo a hidden value $S$, with values $R_1$ and $R_2$ respectively. Such values $R_1$ and $R_2$ are chosen uniformly at random from the hidden ring $\mathbb{Z}_{S}$ with certain conditions. We propose to use said homomorphic encryption in conjunction with Multivariate Polynomial Public-key Cryptography, to secure the polynomial public keys, however, we do not study the encrypted MPKC in detail. Instead, we suggest a new variant of multivariate public-key cryptosystem with public keys encrypted using homomorphic encryption, called Homomorphic Polynomial Public Key or HPPK. We described the HPPK algorithm in detail, with the framework drawn from MPKC. HPPK public keys are product polynomials of a multivariate and univariate polynomials, encrypted with a homomorphic encryption operator. The ciphertext is created by the encrypting party through the input of plaintext and random noise as public polynomial variables. The decryption procedure involves first decrypting the public key, to nullify the homomorphic encryption and produce the original ciphertext. Said ciphertext is used to divide two product polynomials. By construction, such division cancels the base multiplicand polynomial with noise variable. and retains a single equation in one variable. Said variable is the plaintext, which can be found by radicals. We give a thorough security analysis of the HPPK cryptosystem, proving that the hardness of breaking the HPPK algorithm comes from the computational hardness of the Modular Diophantine Equation, and Hilbert's tenth problem. We also show that HPPK holds IND-CPA property. We report briefly on benchmarking the performance of the HPPK cryptosystem, using the NIST-recognized SUPERCOP benchmarking tool, with $n_b=1$ and $\lambda=1$. The benchmarking data illustrates that the HPPK offers rather small public keys and comparable ciphertext sizes. The key generation, key encapsulation, and key decapsulation procedure performance are efficient, being noticeably faster when considered together with the NIST PQC finalists. If the degree of univariate polynomial $f_1(x)$ and $f_2(x)$ are higher than 1, such as quadratic polynomials, the decryption would produce multiple roots, then an extra verification procedure is required. Moreover, the decryption speed would be dramatically slower than linear polynomials $f_1(x)$ and $f_2(x)$. In the future work, we will perform more detail benchmarking with variety of configurations as well as a more extensive security analysis, considering attacks that have not been described in the work.

\bibliographystyle{unsrt}

\end{document}